\documentclass[runningheads,a4paper]{llncs}

\usepackage{graphicx}
\usepackage{amsmath}
\usepackage{amsfonts}
\usepackage{amssymb}

\usepackage[mathscr]{euscript}

\usepackage{amscd}
\usepackage[utf8]{inputenc}
\usepackage{stmaryrd}
\usepackage{xcolor}

\usepackage{url}

\newcommand{\citeonline}[1]{\cite{#1}}

\newtheorem{defn}{Definition}
\newtheorem{prop}[defn]{Proposition}
\newtheorem{theo}[defn]{Theorem}

\newtheorem{fact}[defn]{Fact}

\newcommand{\citey}[1]{\cite{#1}}

\begin{document}

\author{Marlo Souza\inst{1}%
\and Álvaro Moreira\inst{2}}%

\institute{Institute of Mathematics and Statistics, Federal University of Bahia - UFBA,\\
Av. Adhemar de Barros, S/N, Ondina - Salvador-BA, Brazil,\\
\email{marlo@dcc.ufba.br},\\
\and
Institute of Informatics, Federal University of Rio Grande do Sul- UFRGS,\\
Av. Bento Gonçalves, 9500 - Porto Alegre-RS, Brazil\\ 
\email{alvaro.moreira@inf.ufrgs.br}
}

\mainmatter  

\title{Bringing Belief Base Change into Dynamic Epistemic Logic}

\maketitle

\begin{abstract}
AGM's belief revision is one of the main paradigms in the study of belief change operations. In this context, belief bases (prioritised bases) have been primarily used to specify the agent's belief state. While the connection of iterated AGM-like operations and their encoding in dynamic epistemic logics have been studied before, few works considered how well-known postulates from iterated belief revision theory can be characterised by means of belief bases and their counterpart in dynamic epistemic logic. Particularly, it has been shown that some postulates can be characterised through transformations in priority graphs, while others may not be represented that way. This work investigates changes in the semantics of Dynamic Preference Logic that give rise to an appropriate syntactic representation for its models that allow us to represent and reason about iterated belief base change in this logic.
\end{abstract}

\section{Introduction}
\label{sec:intro}

Belief Change is the study of how an epistemic agent comes to change her mind after acquiring new information. While changes in mental attitudes is a well-studied topic in the li\-te\-ra\-tu\-re, 
the integration of such changes as operations within logics of beliefs, obligations, etc. is a somewhat recent development. 


Inspired by the Dutch School, several dynamic logics for information change have been proposed \cite{van2007dynamic,BAL08} which can be connected to the study Belief Change. In particular,  Girard \citey{girard2008modal} proposes Dynamic Preference Logic (DPL) which has been applied to study generalisations of belief revision \textit{a la} AGM \cite{girard2008modal,girard2014belief}. Interestingly, \citeonline{souza:dali} have proposed using DPL as a tool to investigate different classes of belief change operators.

Belief Base Change is the area that studies Belief Change based on syntactic representations of the agent's epistemic commitments. The area arises from Hansson's \citey{Hansson} criticism of the use of deductively closed sets of formulas to represent an agent's epistemic state in the AGM paradigm.Recently, the notion of belief base has been extended into similar (and more expressive) structures such as e-bases \cite{ROT98}, epistemic entrenchments \cite{mreis:igpl14},  priority graphs \cite{liu2011reasoning} and others.

While connections between Belief Base Change and Belief Change have been investigated in the literature \cite{hansson1994kernel,ferme:JLC17,ferme:JSL08}, they rely mainly on one-shot changes, thus, are not able to clarify the behaviour of iterated changes. As such, it is necessary to establish which formal properties these belief base change operations satisfy in a dynamic sense, as studied in the literature of Iterated Belief Change \cite{darwiche,nayak2003dynamic,JIN06}. 

The relationship between priority graphs (and similar structures) and  preference relations has been widely investigated in the literature \cite{andreka2002operators,liu2011reasoning}. Logics relating notions of belief/preference based on syntactic representations, such as priority graphs, and semantic representations, have been studied before by authors such as Levesque~\cite{levesque1984logic}, van Benthem et al.~\cite{liu:deontics} and Lorini~\cite{lorini2018praise}. More yet, it has been shown that several well-known dynamic operators for preference change can be characterised by means of transformations on such structures \cite{andreka2002operators,liu:deontics,souzakr}. 

Souza et al.~\citey{souzaaaai2019} have taken one step in the direction of connecting Belief Base Change and Iterated Belief Change  by  characterising well-known postulates as conditions that must be satisfied by transformations on priority graphs. These authors show, however, that an important class of postulates cannot be represented that way, demonstrating that there is an expressiveness gap between priority graphs and preference models concerning the dynamic aspects of the logic.

In this work, we aim to shorten the expressiveness gap between preference models and priority graphs. We show that for an appropriate restriction of the class of models, we can provide syntactic structures (called grounded priority graphs) able to encode all the information contained in a preference model, i.e. not only the information relative to the ordering among the possible worlds but also extensional information on which worlds are epistemically possible. As such, we can study operations of belief change based on preference models as transformations on grounded priority graphs. As a result,  we obtain the representation of two postulates shown by Souza et al.~\citey{souzaaaai2019} to not be representable in such a manner.  We also obtain a characterisation of relevant priority graph transformations, i.e., transformations that may be used to represent belief change operators, a problem that was left open in previous works, such as \cite{liu2011reasoning,souzakr,souzaaaai2019}. 


This work is structured as follows:  in Section~\ref{sec:ibr}, we present the background theory on Iterated Belief Change; in Section~\ref{sec:dyn}, we present Dynamic Preference Logic, as well as the relevant connections between preference models and priority graphs in the literature. We also prove some fundamental representation results, which strengthen the ones in the literature, based on the definition of preference model we propose. In Section~\ref{sec:iter}, we employ DPL to study Iterated Belief Change operators through transformations on priority graphs. In this section, we provide a characterisation of relevant priority graph transformations, and we use them to characterise important belief change postulates which could not be represented by previous methods. Finally, in Section~\ref{sec:final}, we provide some final considerations, reflecting on the epistemological limitations of our logic, as well as possible future work and applications.

\section{Dynamic and Iterated Belief Change}
\label{sec:ibr}

AGM's work \cite{AGM} focused on postulating minimal requirements for belief change operations in order to describe rational ways of changing one's beliefs.

Among the three basic operations studied by AGM, only expansion can be univocally defined. 
The operations of revision and contraction, on the other hand, are constrained by a set of postulates,  usually referred to as AGM postulates, that define a class of suitable change operators, representing different rational ways in which an agent can change her beliefs. 

It has been argued that AGM's approach lacks a clear semantic interpretation.  Based on Lewis' models for counterfactual reasoning \citey{lewis}, Grove~\citey{grove} provided a possible-world semantics to AGM operations for a (supraclassical and monotonic) logic $\mathcal{L}$.  
He shows that for any belief revision operator $*$ satisfying the AGM postulates and any belief set $B$, there is a system of spheres ${S_B = \langle W, \leq\rangle}$, in which $W$ is a set of models for the language $\mathcal{L}$, such that $w \in \mbox{\textit{Min}}_\leq W$ iff $w\vDash B$ and $\llbracket B*\varphi \rrbracket_{S_B} = \mbox{\textit{Min}}_\leq \llbracket \varphi\rrbracket_{S_B}$. As such, compliance to AGM's postulates can be semantically characterised by postulate \textsc{Faith} below, which states the the minimal worlds on the revised epistemic state of the agent are exactly the minimal worlds satisfying a certain property $\varphi$ that the agent has come to believe to be true,  on changes in Grove's models  \cite{nayak2003dynamic}:
\begin{itemize}
\item[] (\textsc{Faith}) $w \in \textit{Min}_\leq \llbracket \varphi\rrbracket$ iff $w \in \textit{Min}_{\leq_{*\varphi}} W$
\end{itemize}

It has been pointed out that AGM belief revision says very little about how to change one agent's beliefs repeatedly. In fact,  it has been observed that the AGM approach allows some counter-intuitive behaviour in the iterated case \cite{darwiche}. To remedy this deficiency, Darwiche and Pearl \citey{darwiche} propose a set of additional postulates that further constrain the behaviour of revision operators.  
%
%
Furthermore, the authors analyse the proposal by Boutilier~\citey{boutilier1993revision}  of Natural Revision. To model this operation, they propose the postulate of \textit{conditional belief minimisation} \textsc{CB}, which states that the conditional beliefs of the agent (which are not related to the property being revised) are maintained.

\begin{enumerate}
\item[](\textsc{CB}) If $w,w' \not\in \textit{Min}_\leq \llbracket \varphi\rrbracket$, then $w\leq w'$ iff $w \leq_{*\varphi} w'$.
\end{enumerate}

Darwiche and Pearl use this postulate to characterise the operation of Natural Revision showing thus that this is only an example of their broader notion of iterated belief change.

\begin{defn}[Natural Revision]
Let  \mbox{$\leq \,\,  \subseteq  W \times W$} be a plausibility relation and $\varphi$ a propositional formula. The Natural Revision of $\leq$ by $\varphi$ is the plausibility relation \mbox{$\leq_{*\varphi} \,\,   \subseteq W\!\times W$} satisfying \textsc{Faith}  and \textsc{CB}.
\end{defn}

Based on criticism by Freund and Lehman~\citey{lehmanfreund}, Nayak et al. \citey{nayak2003dynamic}, however, show that \textsc{DP} postulates are incompatible with the original AGM postulates. To solve this problem, they  propose the notion of dynamic revision operator, in which a belief revision changes not only the belief set of the agent but the operation itself, i.e., the agent's epistemic state. This distinction between \textit{static} and \textit{dynamic} operators has been observed to be relevant in works such as that of van Benthem \citey{van2007dynamic} and Baltag and Smets \citey{BAL08}, or that of Lindstr{\"o}m and Rabinowicz \citey{lindstrom:rabinowicz}, in which AGM-like static revision can be seen as a counterfactual reasoning while dynamic revision is modelled as an epistemic action changing the agent's epistemic state. 

In this work, we explore the characterisation of Belief Change postulates within Dynamic Preference Logic using both the proof theory of the logic and its characterisation through transformations on belief bases, understood here as priority graphs. To do this, in the following section, we introduce Dynamic Preference Logic, the logic that we will use to reason about Belief Change.

\section{Dynamic Preference Logic}
\label{sec:dyn}

Preference Logic is a modal logic of transitive and reflexive frames. It has been applied to model a plethora of phenomena in Deontic Logic~\cite{liu:deontics}, Epistemic Logic~\cite{BAL08}, etc. Dynamic Preference Logic (DPL) \cite{girard2008modal} is the result of ``dynamifying'' Preference Logic, i.e., extending it with dynamic modalities allowing the study of dynamic phenomena of attitudes such as Beliefs, Obligations, Preferences, etc.

We begin our presentation with the language and semantics of Pre\-fe\-ren\-ce Logic, which we will later ``dynamify''. 
\begin{defn}
Let $P$ be a finite set of propositional letters. We define the language $\mathcal{L}_{\leq}(P)$ by the following grammar (where $p \in P$): $$
\varphi ::= p ~|~ \neg \varphi ~|~ \varphi \wedge \varphi ~|~ A \varphi ~ | ~ [ \leq ] \varphi~ | ~ [<]\varphi$$
\end{defn}

We will often refer to the language $\mathcal{L}_\leq(P)$ simply as $\mathcal{L}_\leq$ by supposing the set $P$ is fixed. Also, we will denote the language of propositional formulas by $\mathcal{L}_0(P)$ or simply $\mathcal{L}_0$. Girard~\citey{girard2008modal} has proposed a semantics for DPL based on Kripke frames with a reflexive and transitive accessibility relation. However, as pointed out earlier, Souza et al.~\citey{souzaaaai2019} show that for this class of models, some belief change operators cannot be represented by means of the manipulation of syntactic representations of these  models, presented later in this section. For this reason, we propose a variation of the notion of preference models that, we show, possess good representational properties from a dynamic point of view.

\begin{defn}\label{def:prefmod}
A conditionally-grounded preference model is a tuple \mbox{$M\!=\!\langle W, \leq, v\rangle$} where $W\subseteq 2^P$ is a set of possible worlds, $\leq$ is a a reflexive and transitive relation over $W$,  and $v: P \rightarrow 2^W$, s.t. $w\in v(p)$ iff $p\in w$, is a valuation function.
\end{defn}

In such a model, the accessibility relation $\leq$ represents an ordering of the possible worlds according to the preferences of a certain agent. As such, given two possible worlds $w,w' \in W$, we say that $w$ is at least as preferred as $w'$ if, and only if, $w\leq w'$. 

Notice that in Definition~\ref{def:prefmod}, we require that $W\subseteq 2^P$, i.e., possible worlds are possible propositional valuations. This requirement has important expressiveness consequences for the logic. Particularly, one expressiveness consequence can be seen in Theorem~\ref{teo:pgraph}, a stronger version of a representation theorem due to Liu~\citey{liu2011reasoning}. 

The choice of restriction in the class of models we present in this work stems from two reasons: these models are more connected with Grove's proposal of models for belief change, in which possible worlds were maximal consistent theories of the logic - in the propositional classical logic case, it is equivalent to valuations - and the fact that, as showed by Andersen et al.~\cite{andersen2017bisimulation}, for conditional logics defined over linear preference model, any (linear) preference model is modally equivalent to a conditionally-grounded (linear) preference model\footnote{The authors consider only linear models in their work and, a priori, it is not clear whether their modal equivalence result can be extended to pre-orders in general. Nevertheless, it indicates that conditionally-grounded models preserve a great deal of conditional information held in general preference models and, as such, constitute an interesting subclass of models to be studied for this logic. Our results in this work only support this conclusion by showing that, for considering this subclass of models, we can obtain interesting representation results that allow computational exploration of DPL in diverse areas.}. From their work, we also know that the logic of degrees of belief, which can be encoded within Preference Logic, is more expressive than the logic of conditional beliefs, which suggests that the restriction on the models has important expressibility consequences for our logic.

Since, in this work, our investigation only concerns those postulates which are defined by means of conditions on Grove models, usually based on the notion of conditional belief, we do not believe this expressive limitation of the logic will be of great concern. We nevertheless discuss these limitations in our Final Considerations. In the following, we will refer to conditionally-grounded preference models simply as preference models, for the sake of the presentation. 

The interpretation of the formulas over these models is defined as usual. Here, we only present the interpretations for the modalities, since the semantics of the propositional connectives is clear. 
They are interpreted as $$\begin{array}{lll}
M,w \vDash A\varphi &\mbox{ iff } &\forall w' \in W:\, M,w'\vDash\varphi\\
M,w \vDash [\leq]\varphi & \mbox{ iff }&\forall w' \in W:\, w'\leq w \Rightarrow M,w' \vDash  \varphi\\
M,w \vDash [<]\varphi & \mbox{ iff }&\forall w' \in W:\, w'< w \Rightarrow M,w' \vDash \varphi
\end{array}$$

As usual, we will refer as $\langle \leq \rangle \varphi$ and $\langle < \rangle \varphi$ to the formulas $\neg [\leq]\neg \varphi$ and  $\neg [<]\neg \varphi$, respectively,  as commonly done in modal logic. Also, given a model $M$ and a formula $\varphi$, we use the notation $\llbracket \varphi\rrbracket_M$ to denote the set of all the worlds in $M$ satisfying $\varphi$. When it is clear which model we are referring to, we will denote the same set by $\llbracket \varphi\rrbracket$. Also, as usual, we will refer as $M\vDash \varphi$ to the fact that for any world $w$ in the model, it holds that $M,w\vDash \varphi$.

As the concept of most preferred worlds satisfying a given formula $\varphi$ will be of great use in modelling some interesting phenomena in this logic, we define a formula encompassing this exact concept.

\begin{defn}\label{def:mu}
We define the formula $\mu\varphi \equiv \varphi \wedge \neg \langle < \rangle \varphi$ that is satisfied by exactly the most preferred worlds satisfying $\varphi$, i.e., $\llbracket \mu \varphi \rrbracket_M = \textit{Min}_\leq \llbracket \varphi\rrbracket_M$.
\end{defn}

\subsection{Preferences and Priorities}
\label{sec:dynamic}

The relation between preference relations and their representations as orderings over formulas (priority orderings, or entrenchment relations) has been extensively studied in the literature \cite{gardenfors1988revisions,georgatos1999preference,andreka2002operators}. Liu~\citey{liu2011reasoning} explore this relationship to propose a syntactic representation of preference models, called priority graphs (or P-graphs for short) which can be used to reason about conditional preferences in DPL \cite{liu:deontics}. With this connection, we will be able to investigate well-known postulates of Iterated Belief Change, defined over preference relations as those in preference models, using operations on priority graphs, thus connecting belief base change operations and iterated belief revision.

\begin{defn} \cite{liu2011reasoning}
Let $P$ be a countable set of propositional symbols and $\mathcal{L}_0(P)$ the language of classical propositional sentences over the set $P$. A priority graph is a tuple \mbox{$G = \langle \Phi, \prec \rangle$} where $\Phi \subset \mathcal{L}_0(P)$, is a finite set of propositional sentences and $\prec$ is a strict partial order on $\Phi$.
\end{defn}

It is easy to see that from a priority graph, we can construct a preference model by taking the preference relation induced by such a graph.

\begin{defn}\label{def:model}
Let $G = \langle \Phi, \prec \rangle$ be a P-graph  and $M = \langle W, \leq, v\rangle$ be a preference model. We say   that $M$ is induced by $G$ iff for any $w,w'\in W$ it holds that
\[
\begin{array}{l}
w \leq w' ~~ \mathit{iff}  ~~~\forall \varphi \in \Phi: ((w' \vDash\varphi \Rightarrow w \vDash\varphi) ~ \mathit{or}\\
\quad\quad\quad\quad\quad\quad\quad\quad\quad\quad    \exists \psi \in \Phi : \psi \prec \varphi, ~ w\vDash \psi,\mathit{and} ~w'\not\vDash\psi)
 \end{array}
 \]
\end{defn}

Clearly, given a P-graph $G = \langle \Phi, \prec \rangle$, the relation $\leq$ satisfying the condition in Definition~\ref{def:model} is reflexive. Notice that it is also transitive since for any words $w,w',w''$, if $w\leq w'$ and $w'\leq w''$ then for any $\varphi''$ that $w''$ satisfies, then either $w'$ satisfies it or there is some $\varphi'$ that $w'$ satisfies and $w''$ doesn't and $\varphi' \prec \varphi$. Similarly, since $w\leq w'$ there is some $\varphi$ that $w$ satisfies and $w'$ doesn't and $\varphi \prec \varphi'$. Notice that, since $\Phi$ is finite, we can take the minimal $\varphi'$ and $\varphi$ for which these properties hold, from which we can conclude that $w''$ cannot satisfy $\varphi$, otherwise either $w \not\leq w'$ or $w'\not\leq w''$. Thus, by transitivity of $\prec$, $w \leq w''$.  

The induction of preference models from P-graphs raises the question about the relations between these two structures. Liu~\citeonline{liu2011reasoning} shows that any preference frame,  i.e., any set of worlds with a reflexive and transitive accessibility relation, is induced by some P-graph. This result cannot be strengthened to preference models, however, since by fixing a certain valuation, it is easy to construct a model for which there is no P-graph that induces it - it suffices to have two worlds $w$ and $w'$ which satisfy the same propositional literals and $w<w'$. Since no propositional formula can distinguish the world $w$ from $w'$, no P-graph can express the order $w<w'$.  Within the class of preference models defined in this work, we can strengthen further this result showing that any conditionally-grounded preference model is induced by a P-graph.

\begin{theo}\label{teo:pgraph}
Any preference model $M = \langle W, \leq,v\rangle$ is induced by a priority graph \mbox{$G=(\Phi,\prec)$}.  
\end{theo}
\begin{proof}
Take $C_M=\{[w] ~|~ [w] = \{w'\in W \mbox{ s.t. }w'\leq w \mbox{ and } w\leq w'\}\}$, we define the characteristic formula of a cluster $[w]$, the formula $\varphi_{[w]} = \bigvee_{w'\in [w]}\bigwedge_{p\in P}p(w')$ s.t. $p(w')=p$ if $M,w'\vDash p$ and $p(w')=\neg p$, otherwise. With that, construct $G_M =\langle \Phi, \prec\rangle$ on the following way.
\begin{itemize}
\item $\Phi = \{\varphi_{[w]} ~|~[w]\in C_M\}$
\item $\varphi \prec \psi$ iff there are $w,w'\in W$ s.t. $M,w\vDash \varphi$ and $M,w'\vDash \psi$ and $w<w'$.
\end{itemize}
Notice that each world $w$ in the model satisfies exactly one formula of $\Phi$ and only worlds in the same cluster, i.e., equally preferable to each other, satisfy the same formula $\Phi$, since the formula $\varphi_{[w]}$ is a disjunction of the characteristic formulas $\varphi_{w'}$ of the worlds in the cluster $[w]$. As such, it clearly holds that for any $w,w'\in W$, $w\leq w'$ iff either  $\varphi_{[w]} = \varphi_{[w']}$ or $\varphi_{[w]}\prec \varphi_{[w']}$. We call $G_M$ the canonical P-graph inducing $M$. \qed
\end{proof}

Notice that not necessarily two P-graphs that induce the same preference model are equal (or isomorphic in some sense), since, as can be easily seen from Definition~\ref{def:model}, any submodel of a preference model induced by some P-graph $G$ is also induced by $G$. As such, preference models are underdetermined, in a sense, by P-graphs and two preference models with substantially distinct canonical models can be induced by the same P-graph.

\begin{fact}\label{fact:submodels}
Let $G = \langle \Phi, \prec\rangle$ be a P-graph and let also $M = \langle W, \leq, v\rangle$ be a preference model induced by $G$. For any preference model $M' = \langle W',\leq',v'\rangle$ s.t. $W'\subseteq W$, $\leq' \, = \,  \leq_{\rvert_{W'}}$ and $v' \, = \, v_{\rvert_{W'}}$, $M'$ is induced by $G$.
\end{fact}

\section{Iterated Belief Change and Dynamic Preference Logic}
\label{sec:iter}

In this section, we investigate the relationship between the postulates satisfied by iterated belief change operators discussed in Section~\ref{sec:ibr} and their characterisation inside Dynamic Preference Logic. 

We define a dynamic operation on a preference model as any operation that takes a preference model and a formula and changes the preference relation of the model. 

\begin{defn}\cite{souza:dali}
We say $\star$ is a dynamic operator on preference models if for any preference model $M= \langle W,\leq,v\rangle$ and formula $\varphi\in \mathcal{L}_0$, we have that  $\star(M,\varphi) = \langle W, \leq_\star, v\rangle$. In other words, an operation on preference models is called a dynamic operator iff it only changes the relation of the preference model. We will use $M_{\star \varphi}$ to denote the model $\star(M,\varphi)$.
\end{defn}

Given a dynamic operator $\star$, we extend the language $\mathcal{L}_\leq$ with formulas $[\star \varphi]\xi$. Here, we point out some abuse of notation, since we use $\star$ as both a dynamic operator defined as a function and as a symbol in the object language to define the modality $[\star \varphi]$  - which will correspond to the application of this operator $\star$ to the model. 

\begin{defn}
Let $\star$ be a dynamic operator. We define the language $\mathcal{L}_\leq (\star)$ as the smallest set containing $\mathcal{L}_\leq$ and all formulas $[\star \varphi]\xi$, with $\varphi \in \mathcal{L}_0$ and $\xi \in \mathcal{L}_\leq (\star)$. 
\end{defn}

Given a preference model $M=\langle W, \leq, v\rangle$, the semantics of formulas $[\star \varphi]\xi$ of  $\mathcal{L}_\leq (\star)$  is as follows   $$ M,w\vDash [\star \varphi]\xi ~~~\mbox{ iff} ~~~  M_{\star \varphi},w\vDash \xi$$

Notice that, in this work, we are only interested in belief changing operators, i.e., those changing the plausibility the agent attributes to each epistemically possible world, not creating any new knowledge about the world\footnote{As helpfully pointed out by one of the reviewers, since our agents are introspective in the sense that agents know about their beliefs, the belief change operations investigated in this work do change the agent's knowledge, but only in the sense that they change their knowledge about their epistemic state, not their knowledge about the world or current state of affairs. This is an important distinction in the class of operations studied.}.

Liu et al.~\cite{liu2011reasoning,liu:deontics} show that some dynamic belief operators can also be described by means of changes in the priority graphs representing the agent's belief base. In the following,  $\mathbb{G}(P)$ denotes the set of all priority graphs constructed over a set $P$ of propositional symbols. 

\begin{defn}
We call a graph transformation any function $\dagger: \mathbb{G}(P)\times \mathcal{L}_0(P) \rightarrow \mathbb{G}(P)$.
\end{defn}

A P-graph transformation is, thus, a transformation in the agent's belief base, as represented by a priority graph.  Since P-graphs and preference models are translatable into one another, it is easy to connect P-graph transformations and dynamic operators as well. 

\begin{defn}\label{def:induced} \cite{souzaaaai2019}
Let $\star$  be a dynamic operator and $\dagger$ be a P-graph transformation. We say $\star$ is induced by $\dagger$ if for any preference model $M$ and any P-graph $G$, if $M$ is induced by $G$ then the preference model $\star(M,\varphi)$ is induced by the \mbox{P-graph} $\dagger(G,\varphi)$, where  $\varphi$ is any propositional formula in $\mathcal{L}_0(P)$,
\end{defn}

Some difficulties may arise in this connection since the relationship between P-graphs and preference models is not univocal, as exemplified by Fact~\ref{fact:submodels}. We will deal with some of these difficulties in this section through the definition of a more suitable syntactic representation of conditionally-grounded preference models in Subsection~\ref{sub:over}. 

Particularly, it is clear that not all P-graph transformations induce dynamic operators. The reason for this is that, since P-graphs are syntactic representations of preferences, different P-graphs may induce the same preference models. As such, if the P-graph transformation changes these equivalent P-graphs in inconsistent ways, no dynamic operator can satisfy the condition of Definition~\ref{def:induced}. 

For example, a graph transformation that changes the graph $p \prec q$ into the graph $p \prec q$ and the graph $p\wedge q \prec p \wedge \neg q \prec \neg p \wedge q \prec \neg p \wedge \neg q$ into $p\wedge q \prec \neg p \wedge q \prec p \wedge \neg q \prec \neg p \wedge \neg q$ cannot induce any dynamic operator since the original graphs are equivalent, i.e., induce the same models, but the resulting graphs are not. As such, Souza et al.~\citey{souzaaaai2019} define the notion of relevant graph transformation.

\begin{defn}\label{def:relevant} \cite{souzaaaai2019}
We say that  a P-graph transformation $\dagger$ is relevant if there is some dynamic operator $\star$ that is induced by it. 
\end{defn}

\subsection{Characterising relevant graph transformations}

While  some earlier representation results by Liu~\citeonline{liu2011reasoning} guarantee the existence of relevant P-graphs, Souza et al.~\citeonline{souzaaaai2019} did not provide means to identify which graph transformations are relevant or not. To provide such conditions, we need to formalise the notion of equivalence between P-graphs discussed above.

\begin{defn}
Let $G_1$ and $G_2$ be two P-graphs, we say $G_1$ and $G_2$ are $\varphi$-equivalent, symbolically $G_1\equiv_\varphi G_2$, iff they induce the same preference models of a formula $\varphi$, i.e., for any preference model $M$, if $M\vDash \varphi$, then $G_1$ induce $M$ iff $G_2$ induce $M$
\end{defn}

The idea of $\varphi$-equivalence is that two graphs induce the same models when restricted to a certain class of models - represented by a formula $\varphi$. Clearly, it holds that two P-graphs are equivalent (in the sense that they induce the same models) when they are $\top$-equivalent. 

Now, by Proposition~\ref{prop:rel} below, we have that not only relevant graph transformations preserve $\varphi$-equivalence between P-graphs, but this is a sufficient condition for a transformation to be relevant. This is a direct consequence of the Fact~\ref{fact:submodels} presented in Section~\ref{sec:dyn}.

\begin{prop}\label{prop:rel}
Let $\dagger: \mathbb{G}(P)\times \mathcal{L}_0(P) \rightarrow \mathbb{G}(P)$ be a P-graph transformation, the following statements are equivalent:
\begin{enumerate}
\item $\dagger$ is relevant;
\item for any P-graphs $G_1,G_2$ and propositional formula $\varphi$, if $G_1\equiv_\psi G_2$ for some  propositional formula $\psi$, then $\dagger(G_1,\varphi)\equiv_\psi \dagger(G_2,\varphi)$.
\end{enumerate} 
\end{prop}
\begin{proof}
Notice that the proof that statement $1$ implies statement $2$ is trivial by Definition~\ref{def:induced}. The other implication follows easily by defining for any model $M =\langle W, \leq,v\rangle$ and propositional formula $\varphi$, the result \mbox{$*(M,\varphi) = M'=\langle W, \leq', v\rangle$} s.t. $\leq'$ is induced by the graph $\dagger(G,\varphi)$, where $G$ is the canonical P-graph inducing $M$. We must show that $\dagger$ induces $\star$, i.e., that for any P-graph inducing $M$, the result of its transformation by $\dagger$ must induce $M'$. 

Take a P-graph $G'$ that induces $M$. Take the formula $$\psi_M = \left(\bigvee_{w\in W} \varphi_w\right) \wedge \left(\bigwedge_{w\in 2^P\setminus W} \neg \varphi_w\right),$$ where given a propositional valuation $w$, $$\psi_w = \left(\bigwedge_{p\in w}p\wedge\bigwedge_{p\in P\setminus w}\neg p\right).$$ Clearly, for any model $M' = \langle W', \leq', v'\rangle$, $M' \vDash \varphi_M$ iff $W'\subseteq W$. As such, by Fact~\ref{fact:submodels}, it is easy to see that $G\equiv_{\psi_M}G'$. Since $\dagger$ preserves $\psi_M$-equivalence, then $\dagger(G,\varphi)\equiv_{\psi_M}\dagger(G',\varphi)$. Since dynamic operators don't change the set of possible worlds, clearly $M'\vDash \psi_M$ and, as such, it must be induced by $G'$.\qed
\end{proof}

While the notion of $\varphi$-equivalence clarifies the necessary behaviour for a graph transformation to be relevant, it is yet not completely clear how to decide whether a transformation preserves $\varphi$-equivalence. Since, by Fact~\ref{fact:submodels}, we know that a P-graph induces all submodels of some model it induces, to verify if two graphs are $\varphi$-equivalent, it suffices to verify if they induce the same model in the limiting case, i.e., the model containing all and only the valuations that satisfy $\varphi$.

\begin{prop}\label{prop:equiv}
Let $G_1$ and $G_2$ be P-graphs and let $\varphi \in \mathcal{L}_0$ be a propositional formula. We have that $G_1 \equiv_\varphi G_2$ iff there is some preference model $M = \langle \llbracket \varphi \rrbracket_0, \leq, v\rangle$, where $\llbracket \varphi \rrbracket_0$ stands for the set of all propositional valuations satisfying $\varphi$, s.t. $M$ is induced both by $G_1$ and $G_2$.
\end{prop} 
\begin{proof}[Sketch of the proof]
It suffices to see that, given the restriction in the class of models considered in Definition~\ref{def:prefmod}, any model is composed by a subset of propositional valuations in $2^P$. As such, for any model $M = \langle W, \leq, v\rangle$ induced by a P-graph $G$, it must be the case that there is a canonical induced model $M_G = \langle 2^P, \leq_G, v_G\rangle$ induced by $M_G$ s.t. $\leq\subseteq \leq_G$ - by Definition~\ref{def:model}. From there, it is easy to see that $G_1 \equiv_\varphi G_2$ iff $M_{G_1} = M_{G_2}$, as defined above, considering only those valuations that satisfy $\varphi$ in the canonical induced model.

\end{proof}

With this result, we have a tool to verify whether a P-graph transformation is relevant, i.e., induces some dynamic operator. Notice that, since the number of preference models for a finite set $P$  of propositional symbol is also finite, and since there is only a finite amount of semantically distinct propositional formulas over this symbol set, verifying if a graph transformation is relevant is decidable (and, in fact, exponential on the number of propositional symbols in $P$).

\subsection{Overcoming P-graphs expressibility gaps}\label{sub:over}

Souza et al.~\citey{souza:dali} show that the proof theory of DPL can be used to characterise Belief Change postulates. More yet, they study which postulates can also be represented by  conditions on graph transformations, such as \textsc{DP}-1. We say a graph transformation satisfies a postulate if all dynamic operators induced by it satisfy this postulate. 

To characterise belief change postulates using graph transformations, Souza et al.~\citey{souzaaaai2019} provide a set of constraints on transformations that guarantee satisfaction to some postulates. 

Unfortunately, the authors also show that some important postulates from Iterated Belief Revision cannot be characterised by P-graph transformations, as is the case of \textsc{CB}\footnote{Other interesting examples have been previously provided by Souza et al.~\citey{souzakr}, showing that some iterated contraction operators cannot be characterised by P-graphs transformations, unless when restricted to a special class of preference models, which they call broad models.}. However, Natural Revision, which satisfies postulates \textsc{Faith} and \textsc{CB}, is definable in DPL.

\begin{fact}\label{fact:CB}\cite{souzaaaai2019}
No relevant P-graph transformation $\dagger: \mathbb{G}(P)\times \mathcal{L}_0(P) \rightarrow \mathbb{G}(P)$ satisfies both \textsc{Faith} and \textsc{CB}.
\end{fact}

The reason for such a result is that some dynamic operators are defined by means of the minimal worlds in the model satisfying some property. It has been shown, however,  that such a property cannot be encoded employing P-graphs \cite{souzakr}. As such, to overcome the expressiveness gap between dynamic operators and graph transformation, we propose a restriction on the class of models in DPL semantics (corresponding to the requirement that $W\subseteq 2^P$ in Definition~\ref{def:prefmod}) and, below, a variation of the notion of priority graph which, together, guarantee that transformations on priority graphs can appropriately characterise these postulates. 

\begin{defn}\label{def:groundedpgraph}
We call grounded P-graph a structure $G = \langle \varphi, \Phi, \prec\rangle$, s.t. $\varphi\in \mathcal{L}_0$ is a consistent propositional formula and $\langle \Phi, \prec\rangle$ is a P-graph. We also say that $G$ is grounded by $\varphi$.
\end{defn}

The main reason for the lack of expressiveness of P-graphs to define dynamic operations is that a P-graph encodes only information on the structure of the accessibility relation of the models induced by it. The extension of such a model, i.e., which worlds are indeed possible, is not encoded within a P-graph. As such, from the structure of the P-graph, it is not possible to determine which worlds are minimal in the induced models. In Definition~\ref{def:groundedpgraph}, we complement a P-graph with the information $\varphi$, which will be used to define exactly which worlds exist in the induced models.

\begin{defn}\label{def:model2}
Let $G = \langle \varphi, \Phi, \prec\rangle$ be a grounded P-graph, we say that the preference model $M= \langle W, \leq, v\rangle$ is induced by $G$ iff $W = \llbracket \varphi\rrbracket_0$ and $\leq$ is induced by the P-graph $\langle \Phi, \prec\rangle$.
\end{defn}  

Notice that, by Definition~\ref{def:model2}, a grounded P-graph induces exactly one preference model in which the possible worlds are exactly the propositional valuations satisfying the formula $\varphi$ which grounds it. With that addition, we can determine, based solely on the structure of a grounded P-graph, what are the minimal worlds in the induced model that satisfy some formula. To construct a graph-based codification for a formula representing such worlds, based on the work of \citeonline{liu:deontics}, we will use the notion of maximal paths in a graph. 

\begin{defn}
\label{eq:mupath}
Let $G = \langle \varphi,\Phi,\prec\rangle$ be a grounded P-graph, $\sigma=\langle \xi_1, \cdots,\xi_n \rangle$ be a maximal chain of nodes in $G$, and let $\psi$ be a propositional formula, we define the formula $\mu_{\sigma} (\psi)$ representing the minimal worlds in the induced model satisfying $\psi$ as $\mu_\sigma (\psi) =\mu_\sigma^n (\psi)$, where:
$$\mu_\sigma^i (\psi) = 
\begin{cases}
\varphi \wedge \psi & \mbox{if }i=0\\
\xi_i\wedge\mu_\sigma^{i-1} (\psi)  &\mbox{if }i>0\mbox{ and }\xi_i\wedge\mu_\sigma^{i-1} (\psi)\not\vdash \bot\\
\mu_\sigma^{i-1} (\psi)  &\mbox{otherwise}
\end{cases}$$ 
\end{defn}

Notice that in Definition~\ref{eq:mupath}, the formula $\mu_sigma^i(\psi)$ corresponds to the maximal (ordered) conjunction of formulas $\xi_j$ in the path $\sigma = \langle \xi_1, \cdots, \xi_n\rangle$ that is consistent with $\varphi$ and $psi$. Since all possible worlds satisfying $\varphi$ exist in any model $M$ induced by a grounded P-graph $G = \langle \varphi,\Phi,\prec\rangle$, if there is a minimal world in $M$ satisfying $\psi$, it must satisfy the formula $\mu_sigma(\psi)$ for some maximal path $\sigma$ in $G$. With that, it is easy to construct a propositional formula from a grounded P-graph that is satisfied exactly by the minimal $\psi$-worlds in the model induced by it.

\begin{prop}\label{prop:minimal}
Let $G = \langle \varphi,\Phi,\prec\rangle$ be a grounded P-graph, let $\Sigma$ be the set of all maximal chains of nodes in $G$,  and let $M = \langle W, \leq, v\rangle$ be the preference model induced by $G$. For any world $w\in W$, it holds that
$$M,w\vDash \mu_G(\psi) ~~ \mbox{ iff }  ~~ w\in \textit{Min}_\leq \llbracket \psi\rrbracket$$
where $\mu_G(\psi) = \bigvee_{\sigma \in \Sigma} \mu_\sigma(\psi)$
\end{prop} 

Proposition~\ref{prop:minimal} shows that, differently then what happens for P-graphs by Fact~30 of \cite{souzakr}, grounded P-graphs completely define the minimal worlds of their induced conditionally-grounded models. Since the aforementioned fact is used to prove that  well-known contraction operators, as well as postulates \textsc{CB} and \textsc{Faith}, cannot be represented by P-graph transformations, if we consider grounded P-graphs and conditionally-grounded P-graphs, we can provide a representation result for these postulates - and other similar postulates characterised by minimal worlds in a model.  

We do point out that, in this work, we do not allow grounded P-graph transformations to change the grounding of the graph, similar to the condition that dynamic operators cannot change the set of possible worlds of the model. Formally we would need to redefine notions such as P-graph transformation to reinforce this restriction. For space constraints, however, we will refrain from doing so.

Finally, we can characterise \textsc{CB} using both the proof theory of DPL and grounded P-graph transformations.

\begin{prop}\label{prop:CB2}
Let $\star$ be a dynamic operator on preference models satisfying \textsc{Faith}. The operator $\star$ satisfies \textsc{CB} if, and only if, the axiom schemata below are valid in  $\mathcal{L}_\leq(\star)$.
\vspace{-0.2cm}
$$
\begin{array}{lcl}
{}[\star\varphi]p &\leftrightarrow& p\\
{}[\star\varphi](\xi \wedge \xi)&\leftrightarrow& [\star\varphi]\xi \wedge [\star\varphi]\xi\\
{}[\star\varphi]\neg \xi&\leftrightarrow&\neg [\star\varphi]\xi\\
{}[\star\varphi]A \xi &\leftrightarrow& A [\star\varphi] \xi\\
{}[\star\varphi][\leq]\xi & \leftrightarrow & A(\mu \varphi\rightarrow [\star \varphi]\xi)\wedge (\neg \mu \varphi \rightarrow [\leq][\star\varphi]\xi)\\
{}[\star\varphi][<]\xi &\leftrightarrow& \mu\varphi \vee \neg \mu \varphi \rightarrow (A ( \mu\varphi \rightarrow  [\star\varphi]\xi) \wedge [<][\star\varphi]\xi)\\
{}[\leq][\star \varphi]\xi &\rightarrow & \neg\mu\varphi \rightarrow ([\star \varphi][\leq](\neg \mu\varphi \rightarrow \xi))\\
{}[<][\star \varphi]\xi &\rightarrow & \neg\mu\varphi \rightarrow ([\star \varphi][<](\neg \mu\varphi \rightarrow \xi))\\
\end{array}$$
\end{prop}
\begin{proof}[Sketch of the Proof:]
The first implication is straight-forward, by showing that each axiom holds for any preference model $M$ and dynamic operation $\star$ satisfying \textsc{CB} and \textsc{Faith}. It suffices to notice that for any world $w$, if it is a minimal $\varphi$-world, it will become a minimal world in the revised model $M_{\star \varphi}$, by \textsc{FAITH}. If it is not, for any world $w'$ s.t. $w' \leq w$, it holds that $w' \leq_\star w$, by \textsc{CB}, and for any minimal $\varphi$-world $w'$, $w' < w$ by \textsc{FAITH}.

The other implication can be easily shown by observing that for any world in a preference model (in fact, for any set of worlds) there is a propositional formula $\xi$ that is satisfied only by this world (the worlds in this set). With that, and the fact that for any propositional formula $\xi$ and dynamic operator $\star$ it holds in DPL that $\xi \equiv [\star \varphi]\xi$, it is easy to use the axioms above to show that if $\star$ satisfies the postulates above, then it satisfies \textsc{CB}  and \textsc{Faith}.
\end{proof}

Postulate \textsc{CB} can be represented by means of grounded P-graph transformations in the following way.

\begin{prop}
Let $\dagger:\mathbb{G}(P)\times\mathcal{L}_0(P) \rightarrow \mathbb{G}(P)$ be a relevant grounded P-graph transformation. $\dagger$ satisfies \textsc{CB} iff for all grounded P-graph $G=\langle \varphi, \Phi, \prec\rangle$ and propositional formula $\psi\in \mathcal{L}_0$, it holds that $\dagger(G,\psi)$ is $\varphi$-equivalent to the grounded P-graph $G' = \langle \varphi, \Phi_\dagger, \prec_\dagger\rangle$ satisfying:
\begin{enumerate}
\item For all $\xi \in \Phi$, there is some $\xi' \in \Phi_\dagger$ s.t.
\begin{enumerate} 
\item $\varphi \wedge \xi \equiv \varphi \wedge \xi'$ and
\item $\forall \alpha' \in \Phi_\dagger$, if $\alpha' \prec_\dagger \xi'$ then
\begin{enumerate}
\item[] $\alpha'\equiv \mu_G(\psi)$ or
\item[] there is $\alpha \in \Phi$ s.t. $ \varphi \wedge \alpha \equiv  \varphi \wedge \alpha'$ and $\alpha\prec \xi$;
\end{enumerate}
\end{enumerate}
\item For all $\xi' \in \Phi_\dagger$, $\xi' \equiv \mu_G(\psi)$ or there is some $\xi \in \Phi$ s.t.
\begin{enumerate}
\item  $ \varphi \wedge \xi \equiv  \varphi \wedge \xi'$ and  
\item   $\forall \alpha \in \Phi$, if $\alpha \prec \xi$ then there is $\alpha' \in \Phi_\dagger$ s.t.
\begin{enumerate}
\item[] $\varphi \wedge \alpha \equiv  \varphi \wedge \alpha'$ and
\item[] $\alpha'\prec_\dagger \xi'$.
\end{enumerate}
\end{enumerate}
\end{enumerate}
\end{prop}     

Similarly, we can characterise \textsc{Faith} using grounded P-graph transformations:

\begin{prop}
Let $\dagger:\mathbb{G}(P)\times\mathcal{L}_0(P) \rightarrow \mathbb{G}(P)$ be a relevant grounded P-graph transformation. $\dagger$ satisfies \textsc{Faith} if for all grounded P-graph $G=\langle \varphi, \Phi, \prec\rangle$ and propositional formula $\psi\in \mathcal{L}_0$, it holds that $\dagger(G,\psi)$ is $\varphi$-equivalent to the grounded P-graph $G' = \langle \varphi, \Phi_\dagger, \prec_\dagger\rangle$ satisfying:
\begin{enumerate}
\item $\mu_G(\psi) \in \Phi_\dagger$;
\item For all $\xi\in \Phi$, $\mu_G(\psi)\prec_\dagger \xi$.
\end{enumerate}
\end{prop} 

Notice that while in this work we were able to provide characterisations of both \textsc{CB} and \textsc{Faith} using grounded P-graphs, giving similar representations to belief contractions postulates \cite{ramachandran2012three} would be considerably more difficult. The reason for this is that contraction postulates describe constraints that are more fine-grained than those described by revision postulates. As a result, commonly, the restrictions imposed by such postulates would be described by properties on the paths in the resulting grounded P-graph. 

Given the space constraints, we will not explore the representation of these operations in this work, but we do point out that a characterisation of Lexicographic Contraction using P-graphs - which works for our models - has been provided by Souza et al.~\citey{souzakr} in their investigation on contraction operations using DPL. As such, this codification can provide clues for a characterisation of the contraction postulates using graph transformations. 

\section{Final Considerations}
\label{sec:final}

This work has investigated changes in the semantics of Dynamic Preference Logic and Priority Graphs to tackle the expressiveness gap for dynamic properties between dynamic operators over preference models and P-graph transformations. As such, this work can be seen as a step further in the attempt to provide a semantic foundation for the study on Relational Belief Change using Dynamic Preference Logic, as done by Girard et al.~\citey{girard2008modal,girard2014belief} and Souza et al.~\citey{souzakr,souza:dali,souzaaaai2019}. 

Notice that the class of models used in this work is not the same classes used in previous works \cite{girard2008modal,souza:dali}. As highlighted before, one of the contributions of this work is precisely investigating an appropriate class of models for DPL that could give rise to good representation results employing (grounded) priority graphs. Notice that the class of models used by us is closely related to Grove's ~\citey{grove} models of AGM Belief Change.

From an epistemological point of view, since possible worlds are interpreted as epistemically possible, our models are capable of representing the notion of an agent having some knowledge \textit{about} the world (what is epistemically necessary) and her beliefs regarding the state of the world. The restriction in Definition~\ref{def:prefmod} that $W\subseteq 2^P$ states that each possible state of affairs is identified about what is true on the observable properties of the world (represented by propositional symbols). This means that the agent cannot conceive two different state of affairs that are phenomenically identical. This may have important implications in the representation power of our logic regarding auto-epistemic phenomena. This fact highlights the importance of investigating introspective phenomena within Belief Revision Theory to understand the expressive limitations of the theory and its postulates.

As future work, we intend to study how our framework connects to the study of Non-Monotonic Belief Change, as studied by \citeonline{casini2018semantic}. Since preference models can be used to define conditional preferences and non-monotonic rules, we believe our semantic framework is ideal for providing a semantic perspective on the work of these authors. This connection is important to understand  reasoning about change based on non-monotonic rules, such as in the case of goal-oriented reasoning in agent programs \cite{vanriemsdijk}.

We point out that, while we focused on the study of belief changing operations, specifically those that do not the change the agent's knowledge about the world, the results obtained here point to the fact that this framework can be used to study more general belief change operations, such as Public Announcements and those studied by Girard and Rott by means of General Dynamic Dynamic Logic programs  \cite{girard2014belief}. It is not clear, however, if  this approach could be connected with the study of more general relation changing operations available in the literature, such as those studied by Areces et al.~\cite{areces2015relation}.

\bibliographystyle{splncs04}
\bibliography{relatorio}

\end{document}